\documentclass{amsart}
\usepackage{float}
\usepackage{ amssymb }

\newtheorem{theorem}{Theorem}[section]
\newtheorem{lemma}[theorem]{Lemma}
\newtheorem{corollary}[theorem]{Corollary}
\newtheorem{proposition}[theorem]{Proposition}
\theoremstyle{definition}
\newtheorem{definition}[theorem]{Definition}
\newtheorem{example}[theorem]{Example}

\theoremstyle{remark}
\newtheorem{remark}[theorem]{Remark}

\numberwithin{equation}{section}

\newcommand{\F}{\mathbb{F}}
\newcommand{\Z}{\mathbb{Z}}

\begin{document}

\title[The homogeneous weight for $R_k$ and QT-codes over $R_k$]{The homogeneous weight for $R_k$, related Gray map and new binary quasicyclic codes}

\author{Bahattin Yildiz}
\author{Ismail Gokhan Kelebek}
\address{Department of Mathematics, Fatih University, 34500
Istanbul, Turkey} \email{byildiz@fatih.edu.tr,
gkelebek@fatih.edu.tr}

\subjclass[2000]{Primary 94B15; Secondary 94B05}

\keywords{homogeneous weights, cyclic codes, quasitwisted codes, quasicyclic codes,
codes over rings}

\begin{abstract}
Using theoretical results about the homogeneous weights for Frobenius rings, we describe the homogeneous weight
for the ring family $R_k$, a recently introduced family of Frobenius rings which have been used extensively in coding theory. We find an associated Gray map for the homogeneous weight using first order Reed-Muller codes and we describe some of the general properties of the images of codes over $R_k$ under this Gray map. We then discuss quasitwisted codes over $R_k$ and their binary images under the homogeneous Gray map. In this way, we find many optimal binary codes which are self-orthogonal and quasicyclic. In particular, we find a substantial number of optimal binary codes that are quasicyclic of index $8$, $16$ and $24$, nearly all of which are new additions to the database of quasicyclic codes kept by Chen.
\end{abstract}

\maketitle



\section{Introduction}
Codes over Frobenius rings make up an important field of study in
the literature of coding theory. An important aspect of Frobenius
rings is that they possess a so-called generating character and it
is possible to define a homogeneous weight on them. For different
ways to describe the homogeneous weight and the related theoretical
results we refer the reader to such works as \cite{Constant},
\cite{Greferath}, \cite{Greferath2} and \cite{Honold}.

Recently, a family of Frobenius rings denoted by $R_k$ have been introduced and studied through different aspects with coding theory. These are finite commutative rings of characteristic $2$ that are non-principal, non-chain when $k\geq 2$ and that generalize such rings as $\F_2+u\F_2$ and $\F_2+u\F_2+v\F_2+uv\F_2$. More on these rings can be found in \cite{Rk}, \cite{cycRk}, \cite{CycR2}, \cite{R2}, etc.

In this work, contrary to the Lee weights used for the ring $R_k$ up to now, we describe the homogeneous weight for $R_k$, using the theoretical results related to homogeneous weights for Frobenius rings, and in particular using a generating character for the ring $R_k$. We then find a Gray map, $\psi_k$, using the first order Reed-Muller codes, allowing us to map codes over $R_k$ of length $n$ to binary linear codes of length $2^{2^k-1}n$ in a distance-preserving way. We show that the $\psi_k$-images of all linear codes over $R_k$ are self-orthogonal binary linear codes for $k\geq 2$. We also describe quasitwisted codes over $R_k$, giving some theoretical results (such as all quasitwisted codes over $R_k$ of odd coindex are quasicyclic) and describe the binary images under the $\psi_k$-map. We also find many optimal binary codes as the $\psi_k$-images of quasitwisted codes over $R_k$ for suitable $k$. These binary images are at the same time self-orthogonal and quasicyclic of certain indices, almost all of which are new additions to the database of known quasicyclic codes kept in \cite{ChenDatabase}.

The rest of the work is organized as follows:
In section 2 we give some of the preliminaries about the ring family $R_k$ and the homogeneous weights on Frobenius rings. In section 3, we find the general form of the homogeneous weight for $R_k$ and we also describe the related Gray map $\psi_k$. Section 4 includes the theoretical discussion on quasitwisted codes over $R_k$. In section 5, we give the numerical results, where we tabulate the optimal binary codes we have obtained from the $\psi_k$-images of quasitwisted codes over $R_k$. We then finish with section 6, where we mention the concluding remarks and possible directions of future research in related areas.

\section{Preliminaries}
\subsection{The rings $R_k$}
The family of rings denoted by $R_k$ have been introduced in \cite{Rk}. Leaving the details of these rings to the aforementioned work, we recall some of the basic properties, the proofs of which can be found in \cite{Rk}. For $k\geq 1$, let
 \begin{equation} \label{defineRk}
R_k = \F_2[u_1,u_2,\dots,u_k] / \langle u_i^2 =0, u_iu_j = u_j u_i \rangle.
\end{equation}
We actually take $R_0=\F_2$, the binary field.
The basis elements of $R_k$ can be viewed, using subsets $A \subseteq \{1,2, \dots, k\}$ by letting
\begin{equation}
u_A := \prod_{i \in A}u_i
\end{equation}
with the convention that $u_{\emptyset} = 1.$
Then any element of $R_k$ can be represented as
\begin{equation} \label{eq4}
\sum_{A \subseteq \{1, \dots, k\}}c_Au_A, \qquad c_A \in \mathbb{F}_2.
\end{equation}

The ring $R_k$ is a  local  ring with maximal
ideal $\langle u_1,u_2,\dots, u_k \rangle$ and $|R_k| = 2^{(2^k)}$.
It is neither a principal ideal ring nor a chain ring when $k\geq 2$, but is a Frobenius ring for all $k\geq 0$.

An element of $R_k$ is a
unit if and only if the coefficient of $u_{\emptyset}$ is 1 and each
unit is also its own inverse.
The following expresses this more accurately:
\begin{equation}\label{unitsquare}
\forall a\in R_k \:\:\:\: a^2 = \left \{
\begin{array}{ll}
1 & \textrm{if $a$ is a unit}
\\
0 & \textrm {otherwise.}
\end{array}\right.
\end{equation}

The following observations follow easily from the structure of these rings and will be needed later:
\begin{lemma}\label{unituv}{\bf (i)} For any $a\in R_k$, we have
$$a\cdot (u_1u_2\cdots u_k) = \left \{
\begin{array}{ll}
0 & \textrm{if $a$ is a non-unit}
\\
u_1u_2\cdots u_k & \textrm {if $a$ is a unit.}
\end{array}\right.$$
\par {\bf (ii)}For any unit $\alpha \in R_k$ and $x\in R_k$, we have
$$\alpha \cdot x = u_1u_2\cdots u_k \Leftrightarrow x = u_1u_2\cdots u_k.$$
\end{lemma}

We denote the set of units of $R_k$ by $\mathcal{U}(R_k)$ and non-units by $\mathcal{D}(R_k)$. It is clear that
\begin{equation}\label{units}
|\mathcal{U}(R_k)| = |\mathcal{D}(R_k)| = 2^{2^k-1} \:\:\: \textrm{and} \:\:\: \mathcal{U}(R_k) = \mathcal{D}(R_k)+1.
\end{equation}

A linear code of length $n$ over $R_k$ is defined to be an
$R_k$-submodule of $R_k^n$.

\subsection{The Homogeneous weight on Frobenius rings}
Contrary to most of the work done on $R_k$, instead of the Lee
weight and the related Gray map, we will be interested in the
homogeneous weight and the related Gray map. Homogeneous weights
were first introduced by Constantinescu and Heise \cite{Constant}.
We may consider these weights as a generalization of the Hamming
weight for finite rings.

\begin{definition}[{\cite[p.~19]{Greferath}}]
A real valued function $w$ on the finite ring $R$ is called a (left)
homogeneous weight, if $w(0)=0$ and the following is true.\newline
\newline
(H1) For all $x,y\in R,$ $Rx=Ry$ implies $w(x)=w(y).$\newline
\newline
(H2) There exists a real number $\gamma $ such that
\begin{equation*}
\underset{y\in R_{x}}{\sum }w(y)=\gamma \left\vert Rx\right\vert
\text{ for all }x\in R\backslash \{0\}.
\end{equation*}
\end{definition}

The number $\gamma $ is the average value of $w$ on $R$, and from
condition (H2) we can deduce that the average value of $w$ is
constant on every non-zero principal ideal of $R.$

Homogeneous weights for Frobenius rings can be described by using
the M\"{o}bius function. For a finite poset $P$, consider the
function $\mu :P\times P\rightarrow \mathbb{C}$ implicitly defined
by $\mu (x,x)=1$ and $\underset{y\leq t\leq x}{\sum }\mu (t,x)=0$ if
$y<x$ and $\mu (y,x)=0$ if $y\nleqslant x$. It is called the
M\"{o}bius function on $P$ and induces for arbitrary pairs of
real-valued functions $f,g$ on $P$ the following equivalences,
referred to as M\"{o}bius inversion
\begin{equation*}
g(x)=\underset{y\leq x}{\sum }f(y)\text{ for all }x\in P\Leftrightarrow f(x)=%
\underset{y\leq x}{\sum }g(y)\mu (y,x).
\end{equation*}
Let $R$ be a finite ring and $\mu $ be the M\"{o}bius function on the set $%
\left\{ Rx\left\vert x\in R\right. \right\} $ of its principal left
ideals (partially ordered by inclusion). Further let $R^{\times }$
denote the set of units in $R.$ The conditions for the existence and
uniqueness of homogeneous weights on finite rings are given by the
following theorem in \cite{Greferath}.

\begin{theorem}[{\cite[p.~19]{Greferath}}]
\label{weight} A real valued function $w$ on the finite ring $R$ is
a homogeneous weight if
and only if there exists a real number $\gamma $ such that $w(x)=\gamma %
\left[ 1-\frac{\mu (0,Rx)}{\left\vert R^{\times }x\right\vert
}\right] $ for all $x\in R.$
\end{theorem}

Honold \cite{Honold} described the homogeneous weights on Frobenius
rings in terms of generating characters.

\begin{proposition}[{\cite[p.~412]{Honold}}]\label{charac}
Let $R$ be a finite ring with generating character $\chi $. Then
every
homogeneous weight on $R$ is of the form%
\begin{equation}
w:R\rightarrow
\mathbb{R}
,\text{ }x\mapsto \gamma \left[ 1-\frac{1}{\left\vert R^{\times
}\right\vert }\underset{u\in R^{\times }}{\sum }\chi (xu)\right] .
\label{chi}
\end{equation}%
\qquad
\end{proposition}

By Property (H2), the average weight of a left principal ideal of $R$ is $%
\gamma .$ The following proposition shows that for any coset of
either a left or a right ideal, the average weight is $\gamma .$
This property is equivalent to $R$ being Frobenius.

\begin{proposition}[{\cite[p.~412]{Greferath2}}]
Let $I$ be either a left or a right ideal of a finite Frobenius ring
$R$, and let $y\in R.$ Then $\underset{r\in I+y}{\sum }w(r)=\gamma
\left\vert I\right\vert .$
\end{proposition}

\section{The Homogeneous weight and the related Gray map for $R_k$}
\subsection{The Homogeneous weight for $R_k$}
We will find the homogeneous weight for the ring family $R_k$ using
Proposition \ref{charac}. We recall from \cite{Rk} that the
following is a generating character for the Frobenius rings $R_k$:
\begin{equation}
\chi(\sum_{A \subseteq \{ 1,2,\dots,k \}} c_A u_A)  =
(-1)^{wt(\overline{c})},
\end{equation}
where by $wt(\overline{c})$, we mean the Hamming weight of the $\F_2$-coordinate vector of the element in the basis $\{u_A|A \subseteq \{1,2,\dots, k\}\}$.

For example, for the case of $k=2$, we have
$\chi (0)=1$

$\chi (1)=\chi (u)=\chi (v)=\chi (uv)=-1$

$\chi (1+u)=\chi (1+v)=\chi (1+uv)=\chi (u+v)=\chi (u+uv)=\chi
(v+uv)=1$

$\chi (1+u+v)=\chi (1+u+uv)=\chi (1+v+uv)=\chi (u+v+uv)=-1$

$\chi (1+u+v+uv)=1$.

The following lemma will be a key in proving the main theorem about
the homogeneous weight on $R_k$:
\begin{lemma}\label{lem}
Let $x$ be any element in $R_k$ such that $x\neq 0$ and $x\neq u_1u_2\cdots u_k$. Then
$$\sum_{\alpha \in \mathcal{U}(R_k)}\chi(\alpha x) = 0.$$
\end{lemma}
\begin{proof}
Since $\chi$ is a generating character, it is non-trivial when restricted to any non-zero ideal, and thus we have,
\begin{equation}
\sum_{\alpha \in R_k}\chi(\alpha x) = 0.
\end{equation}
But, by \ref{units}, since $\alpha \in \mathcal{U}(R_k)$ if and only if $1+\alpha \in \mathcal{D}(R_k)$, and
$$\chi((\alpha+1)x) = \chi(\alpha x+x) = \chi(\alpha x)\cdot \chi(x),$$
the above sum becomes
\begin{equation}\label{eq}
0 = \sum_{\alpha \in \mathcal{U}(R_k)}\chi(\alpha x)+\sum_{\alpha \in\mathcal{R_k}}\chi(\alpha x)\chi(x) = (1+\chi(x))\sum_{\alpha \in \mathcal{U}(R_k)}\chi(\alpha x).
\end{equation}
The proof is done if $\chi(x) = 1$. Now, assume $\chi(x) = -1$. Let us label the sum:
$$F(x) = \sum_{\alpha \in \mathcal{U}(R_k)}\chi(\alpha x).$$
As $\alpha$ runs through all the units of $R_k$, we can easily observe that
$F(x) = F(\beta x)$ for all $\beta \in \mathcal{U}(R_k)$.
Thus (\ref{eq}) can be written as
\begin{equation}
(1+\chi(\alpha x))F(x) = 0, \:\:\: \forall \alpha \in \mathcal{U}(R_k).
\end{equation}
So, the proof will be complete if we prove that $\chi(\alpha x) =1$ for at least one value of $\alpha \in \mathcal{U}(R_k)$. Assume to the contrary that $\chi(\alpha x) = -1$ for all $\alpha \in \mathcal{U}(R_k)$. But, then this means we must have $\chi(\beta x) = 1$ for all $\beta \in \mathcal{D}(R_k)$. Now, the ideal generated by $x$ must contain $u_1u_2 \dots u_k$. Since $x \neq u_1u_2 \cdots u_k$, $\alpha x \neq u_1u_2 \dots u_k$
for any $\alpha \in \mathcal{U}(R_k)$ by Lemma \ref{unituv}. Thus we must have $\beta x = u_1u_2 \dots u_k$ for some $\beta \in \mathcal{D}(R_k)$. But this is a contradiction since $\chi(u_1u_2 \dots u_k) = -1$.
\end{proof}

We are now ready to describe the homogeneous weight for $R_k$:
\begin{theorem}
The homogeneous weight on $R_k$ is found to be:
\begin{equation*}
w_{\hom }(x)=\left\{
\begin{array}{c}
0 \\
2\gamma  \\
\gamma
\end{array}%
\begin{array}{c}
\text{if}~x=0 \\
\text{if }x=u_1u_2 \dots u_k \\
\text{otherwise.}%
\end{array}%
\right.
\end{equation*}%
\end{theorem}
\begin{proof}
Suppose $x = u_1u_2 \dots u_k$. Then by Lemma \ref{unituv}, $\alpha x = x$ for all $\alpha \in \mathcal{U}(R_k)$. Thus $\chi(\alpha x) = -1$ for all  $\alpha \in \mathcal{U}(R_k)$. Hence, by Proposition \ref{charac}, we have
$$w_{\hom}(x) = \gamma \left[ 1-\frac{1}{\left\vert \mathcal{U}(R_k)\right \vert }\underset{\alpha\in \mathcal{U}(R_k)}{\sum }(-1)\right] = 2\gamma.$$

If $x\neq 0$ and $x\neq u_1u_2 \dots u_k$, then by Lemma \ref{lem}, we have $\sum_{\alpha \in R_k}\chi(\alpha x) = 0$. Thus
we obtain
$$w_{\hom}(x) = \gamma \left[ 1-\frac{1}{\left\vert \mathcal{U}(R_k)\right\vert}  0 \right] = \gamma.$$
\end{proof}

\subsection{The Gray map for the homogeneous weights}
We will assign a numerical value to $\gamma$ so that it is possible to define a distance-preserving Gray map. In \cite{Pasa}, hyperplanes in projective geometries were used to define a Gray map for the homogeneous weight on $R_k$. Consequently, the choice of $\gamma$ that was imposed by the combinatorial structure was found to be $2^{2^{k}-2}$. A corresponding Gray map from $R_k$ to $\F_2^{2^{2^k-1}}$ was described using the hyperplanes in $PG_{2^{k}-1}(\mathbb{F}_{2})$.

We will adopt the same choice of $\gamma$, but instead of the combinatorial description of the Gray map we will give an algebraic description which uses a well-known family of binary codes, namely first order Reed-Muller codes.

Thus for us the homogeneous weight on $R_k$ is now defined as follows:
\begin{equation}
w_{\hom }(x)=\left\{
\begin{array}{c}
0 \\
2^{2^k-1} \\
2^{2^k-2}
\end{array}
\begin{array}{c}
\text{if}~x=0 \\
\text{if }x=u_1u_2\cdots u_k \\
\text{otherwise.}
\end{array}
\right.
\end{equation}

Let us recall that first order Reed-Muller codes, denoted by $RM(1,m)$ have the following well-known properties:
\begin{enumerate}
\item $RM(1,m)$ is a binary linear code of length $2^m$ and of dimension $m+1$. It contains the all 1-vector of length $2^m$.
\item The minimum weight of $RM(1,m)$ is $2^{m-1}$.
\item Every nonzero codeword other than the all 1-vector has weight $2^{m-1}$. In other words the weight enumerator of $RM(1,m)$ is given by $1+(2^{m+1}-2)z^{2^{m-1}}+z^{2^m}$.
\end{enumerate}

Notice that $R_k$ can be viewed as an $\F_2$-vector space with a basis consisting of $\{u_A| A\subseteq \{1,2, \dots, k\}\}$. The basis has $2^k$ elements. Then if we take $RM(1,2^k-1)$, and define an $\F_2$-linear map $\psi_k: R_k \rightarrow RM(1,2^k-1)$ in such a way that the basis elements are mapped to the basic generating vectors of $RM(1,2^k-1)$, the map $\psi$ will satisfy the following property:
\begin{theorem}
The map $\psi_k$ defined above is a distance preserving isometry from $(R_k, \:\: homogeneous \:\: distance)$ to $(\F_2^{2^{2^k-1}}, \:\: Hamming \:\: Distance)$.
\end{theorem}
\begin{example}
For $k=1$, we get the ring to be $R_1 = \F_2+u\F_2$ and in this case the homogeneous weight coincides with the usual Lee weight that was defined in \cite{Rk}. The map $\psi$ is the usual Gray map defined in the same work, namely $\psi_1(0) = (0,0)$, $\psi_1(1) = (0,1)$, $\psi_1(u) = (1,1)$ and $\psi_1(1+u) = (1,0)$.
\end{example}
\begin{example}
For the case when $k=2$, we describe $R_2$ as $\F_2+u\F_2+v\F_2+uv\F_2$, to go by the original notation used in \cite{R2}. The homogeneous weight
is given by
\begin{equation}
w_{\hom }(x)=\left\{
\begin{array}{c}
0 \\
8\\
4
\end{array}
\begin{array}{l}
\text{if}~x=0 \\
\text{if }x=uv \\
\text{otherwise.}
\end{array}
\right.
\end{equation}
\end{example}
The map $\psi_2$ in this case can be described by assigning the basis elements as follows:
$$\psi_2(uv) = (1,1,1,1,1,1,1,1), \:\:\:\: \psi_2(u) = (1,1,1,1,0,0,0,0)$$
$$\psi_2(v) = (1,1,0,0,1,1,0,0), \:\:\:\: \psi_2(1) = (1,0,1,0,1,0,1,0).$$

The maps $\psi_k$ are naturally extended (component-wise) to $R_k^n$ as well. This allows us to consider the $\psi_k$-images of codes over $R_k$ as well.
Thus we have the following theorem:
\begin{theorem}
Let $C$ be a linear code over $R_k$ of length $n$. Then $\psi_k(C)$ is a binary linear code of length $2^{2^k-1}n$. Moreover the homogeneous weight distribution of $C$ is the same as the Hamming weight distribution of $\psi_k(C)$.
\end{theorem}

Note that when $k\geq 2$, all the homogeneous weights are divisible by $4$. Considering the binary images we get the following observation:
\begin{theorem}\label{self-orth}
Let $C$ be any linear code over $R_k$ of length $n$ with $k\geq 2$. Then $\psi_k(C)$ is a self-orthogonal binary linear code of length $2^{2^k-1}n$.
\end{theorem}

\begin{remark}
The size of a linear code over $R_k$ of length $n$ is at most $2^{2^kn}$. An inductive argument shows that $k+1 < 2^k-1$ for all $k\geq 3$. Thus, for $k\geq 3$, if $C$ is a linear code over $R_k$ of length $n$, then
\begin{align*}
|\psi_k(C)| & \leq 2^{2^kn} = 2^{2^{k+1}.n/2} \\
& < 2^{2^{2^{k-1}}.n/2} = \large \sqrt{|\F_2^{2^{2^k-1}.n}|}.
\end{align*}
This shows that the $\psi_k$-image of a code over $R_k$ cannot be self-dual for any $k\geq 3$. On the other hand, since
$$|R_2^n| = 2^{4n} = \sqrt{2^{8n}}$$ and by Theorem \ref{self-orth}, $\psi_2(R_2^n)$ is self-orthogonal, we see that
$\psi_2(R_2^n)$ is a Type II (i.e. all weights are divisible by 4) binary self-dual code of length $n$ and minimum distance $4$  for all $n\geq 1$. When $n=1$, we get a $[16,8,4]$ extremal Type II binary self-dual code.
\end{remark}

\section{Quasi-twisted codes and their images}
We describe a special class of codes over $R_k$ with respect to the homogeneous distance. Cyclic codes are a special class of codes, which possess an algebraic structure that allows them to be encoded and decoded easily and it also provides more information about the code. Cyclic codes over $R_2$ were studied in \cite{CycR2} and later in more generality in \cite{cycRk} with respect to the Lee metric.

Cyclic codes have two possible generalizations. One is obtained by replacing the ordinary shift with a $\lambda$-shift, which results in $\lambda$-constacyclic codes. The other possible generalization is achieved through replacing the shift with a composition of the shift, thus giving rise to the so-called quasicyclic codes. Constacyclic and quasicyclic codes over $R_2$ were studied in \cite{constaR2} and \cite{QCR2} respectively with respect to the Lee metric. Recently, the concepts of constacyclic and quaciyclic codes have been combined to give rise to a new generalization that generalizes all these concepts, and is called quasitwisted. Quasitwisted codes have been studied for their help in finding many good codes over different alphabets. For the rest of the paper, quasitwisted code will shortly be denoted as QT codes while quasicyclic codes will be denoted by QC codes.

\begin{definition}
Let $R$ be a commutative ring with identity and suppose that $\lambda\in R$ is a unit. A $\lambda$-shift on $R^n$ is defined to be the map $T_{\lambda}$ with the property:
$$T_{\lambda}(a_0,a_1, \dots, a_{n-1}) = (\lambda a_{n-1},a_0,a_1, \dots, a_{n-2}).$$
\end{definition}
When $\lambda=1$ we simply denote it by $T$ and we mean it to be the cyclic shift. This leads to the following definition on codes over $R$:
\begin{definition}
Let $C$ be a linear code over $R$ of length $n$. We say, $C$ is a cyclic code if $T(C) = C$, a $\lambda$-constacyclic code if $T_{\lambda}(C) = C$. For $\ell|n$, we say $C$ is an $\ell$-QC code if $T^{\ell}(C) = C$, and it is a $(\lambda,\ell)$-QT code if $T_{\lambda}^{\ell}(C) = C.$ $\ell$ is called the index of QT codes while $n/\ell$ is referred to as the co-index of the QT codes.
\end{definition}
Structural properties of QT codes over finite fields have been given in \cite{Ackerman}, \cite{Aydin} and \cite{Chen}. We recall that to any vector $(a_0,a_1, \dots, a_{n-1}) \in R^n$, we can assign a polynomial $a_0+a_1x+\dots a_{n-1}x^{n-1} \in R[x]$. This correspondence leads to the following well-known theorems:
\begin{theorem}
{\bf (i)} $C$ is a $\lambda$-constacyclic code over $R$ of length $n$ if and only if the polynomial correspondence of $C$ is an ideal in $R[x]/(x^n-\lambda)$.
\par {\bf (ii)} Suppose $n = \ell \cdot m $. Then a $(\lambda, \ell)$-QT code over $R$ of length $n$ algebraically is an $\left (R[x]/(x^m-\lambda)\right )$-submodule of $\left (R[x]/(x^m-\lambda)\right )^{\ell}$.
\end{theorem}

The following theorem, a special case of which was proved in \cite{constaR2}, can easily be proved:
\begin{theorem}
Let $R$ be a finite commutative ring with identity and let $\lambda$ be a unit in $R$ with $\lambda^2=1$. If $n$ is odd, the map
$$\mu: R[x]/(x^n-1) \rightarrow R[x]/(x^n-\lambda)$$
given by $\mu(f(x)) = f(\lambda x)$ is a ring isomorphism.
\end{theorem}
Since in $R_k$, all units $\lambda$, satisfy $\lambda^2=1$, we get the following corollary:
\begin{corollary}
Let $n$ be odd and $\lambda \in R_k$ be any unit. Then $C$ is a $\lambda$-constacyclic code over $R_k$ of length $n$ if and only if $C$ is a cyclic code over $R_k$ of length $n$.
\end{corollary}
Considering the algebraic structure of the QT codes and the above ring isomorphism, we also get the following result about QT codes over $R_k$:
\begin{corollary}
Let $n =\ell \cdot m$, where $m$ is odd and $\lambda \in R_k$ is any unit. Then $C$ is a $(\lambda,\ell)$-QT code over $R$ if and only if $C$ is an $\ell$-QC code.
\end{corollary}
The above corollaries imply that we do not need to consider constacyclic codes over $R_k$ of odd lengths and $QT$-codes of odd coindex as they are cyclic and QC respectively.

\subsection{The Binary Images of QT-codes over $R_k$}
We start by observing that
$$\psi_k(a_0,a_1, \dots, a_{n-1}) = (\psi_k(a_0), \psi_k(a_1), \dots, \psi_k(a_{n-1}))$$ for all $a_i \in R_k$. Note that $\psi(a_i)$ is a binary vector of size $2^{2^k-1}$. Thus we have
\begin{align*}
\psi_k \circ T(a_0,a_1, \dots, a_{n-1})& = \psi_k(a_{n-1}, a_0, a_1, \dots, a_{n-2}) \\
& = (\psi_k(a_{n-1}), \psi_k(a_0), \dots, \psi_k(a_{n-2})) \\
& = T^{2^{2^k-1}} \circ \psi_k (a_0,a_1, \dots, a_{n-1}).
\end{align*}
In other words, we have
\begin{equation}\label{cycshiftimage}
\psi_k \circ T = T^{2^{2^k-1}}\circ \psi_k.
\end{equation}
This leads to the following theorem:
\begin{theorem}\label{psiImageQC}
If $C$ is a cyclic code over $R_k$ of length $n$, then $\psi_k(C)$ is a binary $2^{2^k-1}$-QC code of length $2^{2^k-1}n$. If $C$ is an $\ell$-QC code over $R_k$ of length $n$, then $\psi_k(C)$ is a binary $(2^{2^k-1}\cdot \ell)$-QC code of length $2^{2^k-1}n$.
\end{theorem}

Now, let $\lambda$ be any unit in $R_k$. Since $\lambda (u_1u_2 \cdots u_k) = u_1u_2 \cdots u_k $ by Lemma \ref{unituv}, we have $w_{\hom}(a) = w_{\hom}(\lambda\cdot a)$ for all units $\lambda$ and elements $a$ in $R_k$. But this means that for any unit $\lambda \in R_k$ and any element $a\in R_k$, $\psi_k(\lambda \cdot a)$ is permutation equivalent to $\psi_k(a)$. This means that equation \ref{cycshiftimage} will have the following form
\begin{equation}
\psi_k \circ T_{\lambda}(a_0,a_1, \dots, a_{n-1}) \simeq T^{2^{2^k-1}}\circ \psi_k (a_0,a_1, \dots, a_{n-1}),
\end{equation}
where $\simeq$ denotes permutation-equivalence.
Thus we have the following version of Theorem \ref{psiImageQC}:
\begin{theorem}
If $C$ is a $\lambda$-constacyclic code over $R_k$ of length $n$, then $\psi_k(C)$ is equivalent to a binary $2^{2^k-1}$-QC code of length $2^{2^k-1}n$. If $C$ is a $(\lambda,\ell)$-QT code over $R_k$ of length $n$, then $\psi_k(C)$ is equivalent to a a binary $(2^{2^k-1}\cdot \ell)$-QC code of length $2^{2^k-1}n$.
\end{theorem}

In the case of $R_2$, which is the most common case we will use in our constructions, we get the following corollary:
\begin{corollary}
Let $C$ be an  $\ell$-QC code of length $n$ over $R_2$. Then $\psi_2(C)$ is a binary $8\ell$-QC code of length $8n$. If $C$ is a $(\lambda,\ell)$-QT code of length $n$ over $R_2$, then $\psi_2(C)$ is permutation-equivalent to a binary $8\ell$-QC  code of length $8n$.
\end{corollary}

\subsection{One-Generator QT codes}
QT codes are structurally complex codes and as such, in their literature the most common types of such codes that have been considered are the so-called one-generator QT codes. The QC, cyclic and constacyclic analogues can easily be considered. Assume that $g(x) \in R[x]/(x^m-\lambda)$ is a polynomial, with $\lambda$ a unit. Then the one-gerenator $\lambda$-constacyclic code generated by $g(x)$ is simply defined to be the principal ideal $\langle g(x)\rangle$ in the ring $R[x]/(x^m-\lambda)$. It is clear that such a code will be generated by the following matrix:
$$G = \left[
\begin{array}{ccccc}
g_{0} & g_{1} & g_{2} & \cdots & g_{m-1} \\
\lambda g_{m-1} & g_{0} & g_{1} & \cdots & g_{m-2} \\
\lambda g_{m-2} & \lambda g_{m-1} & g_{0} & \cdots & g_{m-3} \\
\vdots & \vdots & \vdots & \ddots & \vdots \\
\lambda g_{1} & \lambda g_{2} & \lambda g_{3} & \cdots & g_{0}%
\end{array}%
\right],$$
where $g(x) = g_0+g_1x+ \cdots g_{m-1}x^{m-1}$. In some contexts such a matrix is called a $\lambda$-twistulant matrix or $\lambda$-circulant matrix. When $\lambda=1$, we simply get a circulant matrix as the generator matrix of a one-generator cyclic code.

\begin{definition}
A one-generator $(\lambda, \ell)$-QT code over $R$ is a linear code over $R$ generated by a matrix of the form
$$[G_1|G_2| \cdots |G_{\ell}],$$
where each $G_i$ is an $m\times m$ $\lambda$-twistulant matrix.
\end{definition}
The following theorem, whose cyclic analogue is proved in \cite{QCR2}, is easily obtained.
\begin{theorem}\label{constacyclic}
Let  $C=\langle g(x) \rangle$ be a $\lambda$-constacyclic code over $R_k$ of length $m$ where $g(x)$ is a monic polynomial with  $\deg(g(x))=m-k$. Then $C$ is a free module of rank $k$  if and only if $g(x)|x^m-\lambda$.
\end{theorem}
This theorem can then be used to obtain the following result for a special type of one-generator QT codes:
\begin{theorem}
Suppose $C$ is a $(\lambda, \ell)$-QT code of length $n=m\ell$ generated by $(f_1(x)g(x), f_2(x))g(x), \cdots, f_{\ell}(x)g(x))$, where $x^m-\lambda= g(x) h(x)$ with $g(x)$ and $h(x)$ monic polynomials in $R_k[x]/(x^m-\lambda)$ and $f_i(x)$ is relatively prime to $h(x)$ for all $i=1,2, \dots, \ell$. Then $C$ is a free module with rank $m-deg(g(x))$. In other words, $\psi_k(C)$ is of dimension $2^k(m-deg(g(x))$.
\end{theorem}

There is a natural projection from $R_k$ to its residue field, namely, $\F_2$, which we denote by $\mu_k$. $\mu_k$ essentially works as reduction modulo the maximal ideal, that is
\begin{equation}
\mu_k(\sum_{A \subseteq \{1, \dots, k\}}c_Au_A) = c_{\emptyset}.
\end{equation}
Non-units are sent to $0$, while units are mapped to $1$.

The following lemma provides a natural interval for the minimum homogeneous weight of a code over $R_k$:
\begin{lemma}
Let $C$ be a linear code over $R_k$ and suppose the minimum Hamming weight of $\mu_k(C)$ is $d$. Then
$$2^{2^k-2}d \leq d_{hom}(C) \leq 2^{2^k-1}d.$$
\end{lemma}
\begin{proof}
For any codeword $\overline{c} = (c_1, c_2, \dots, c_m) \in C$, the projection under $\mu_k$ has at least $d$ non-zero coordinates, which means $\overline{c}$ has at least $d$ unit coordinates, all non-zero. Thus the homogeneous weight of $\overline{c}$ is at least $d \cdot 2^{2^k-2}$, giving us the left hand inequality.

For the upper bound, suppose $(a_1,a_2, \dots, a_m) \in \mu_k(C)$ is a binary codeword in $\mu_k(C)$ whose Hamming weight is $d$. Since $\mu_k$ maps units to $1$ and non-units to $0$, this means there exists $(c_1,c_2, \dots, c_m) \in C$ such that exactly $d$ of the $c_i$s are units. But then $u_1u_2\cdots u_k(c_1,c_2, \dots, c_m) \in C$ as well, since $C$ is linear over $R_k$ and this last codeword has exactly $d$ coordinates that are equal to $u_1u_2\dots u_k$ and the rest equal to 0. Since the homogeneous weight of $u_1u_2\dots u_k$ is equal to $2^{2^k-1}$, we see that
$$w_{hom}(u_1u_2\cdots u_k(c_1,c_2, \dots, c_m)) = 2^{2^k-1}d,$$
giving us the necessary upper bound.
\end{proof}
\begin{corollary}
Let $C = ( g_1(x), g_2(x), \dots, g_{\ell}(x)) $ be a one generator $(\lambda, \ell)$-QT code over $R_k$ of length $n= m.\ell$, and suppose the number of unit coefficients of $g_i(x)$ is $d_i$ for $i=1,2,\dots, \ell$. Then
$$d_{hom}(C) \leq 2^{2^k-1}(d_1+d_2+ \dots+ d_{\ell}).$$
\end{corollary}

\section{Examples of optimal binary codes from $\psi_k$-images of QT codes over $R_k$}
In this section we will be giving some examples of optimal codes that we have obtained from the $\psi_k$-images of QT-codes over $R_k$. The optimality of these codes have been established by theoretical upper bounds and specific constructions in \cite{Grassl}. It turns out that many of our constructions serve as alternative (and usually much simpler) constructions for the optimal codes. Because the images under $\psi$ of our codes are binary QC codes of high index, we also compare our results with those collected in \cite{ChenDatabase}, the database of known QT codes. It turns out that we have found many new additions to this database through our constructions.

Before we proceed, we would like to observe that, $\psi_k$ is from $R_k$ to $\F_2^{2^{2^k-1}}$. Thus for example for $k=3$, $\psi_3$ is from $R_3$ to $\F_2^{128}$ while for $k=4$, $\psi_4$ is from $R_4$ to $\F_2^{2^{15}}$. So, for practical purposes, it is not feasible to look beyond $k=2$. That is why, in what follows we will first write down some of the general examples, but then devote separate subsections to the feasible cases of $k=1$ and $k=2$.

\subsection{The $\psi_k$-images of the repetition code}
Let $C$ be the code of length $n$ over $R_k$ generated by $(111 \dots 1)$. It is clear that $C$ is a cyclic and QC code for any suitable index. $C$ is a free code of free rank 1. Thus $|C| = 2^{2^k}$. Considering the homogeneous weights and the $\psi_k$-image we get the following family of binary codes:
\begin{theorem}\label{repet}
Let $C$ be the code of length $n$ generated by $(11\cdots 1)$ over $R_k$. Then $\psi_k(C)$ is a binary linear code
of parameters $[n\cdot 2^{2^k-1}, 2^k, n\cdot 2^{2^k-2}]$. Moreover, when $k\geq 2$, these codes are all self-orthogonal binary codes.
\end{theorem}

\begin{example}
Putting $k=2$ into Theorem \ref{repet}, we get self-orthogonal binary linear codes of parameters $[8n,4,4n]$ from the $\psi_2$-images of the repetition code over $R_2$. For $n=1$ up to $6$, we get the self-orthogonal binary codes of parameters $[8,4,4]$, $[16,4,8]$, $[24,4,12]$, $[32,4,16]$, $[40,4,20]$ and $[48,4,24]$,
all of which are optimal linear codes according to \cite{Grassl}. When $n=7$ and $8$ we get self-orthogonal binary linear codes of parameters $[56,4,28]$ and $[64,4,32]$. The optimal codes of these lengths have parameters $[56,4,29]$ and $[64,4,33]$. But since these latter codes cannot be self-orthogonal, the codes we obtain, have the best possible minimum distance among all the self-orthogonal codes of those lengths and dimensions.
\end{example}

\begin{example}
Putting $k=3$ into Theorem \ref{repet}, we get self-orthogonal binary linear codes of parameters $[128n,8,64n]$ from the $\psi_3$-images of the repetition code over $R_3$. The only cases for which we can make comparisons are the cases when $n=1$ and $n=2$. With these values, we obtain self-orthogonal binary linear codes of parameters $[128,8,64]$ and $[256,8,128]$, both of which are optimal as linear codes.
\end{example}

\subsection{Optimal binary codes from $(1+u,3)$-QT codes over $R_1$}
Note that the ring $R_1 = \F_2+u\F_2$, the first example of the rings we study, has been studied already in the literature for cyclic, quasicyclic and constacyclic codes. We may refer the reader to \cite{Alrub1}, \cite{Alrub2}, \cite{Qian} and \cite{QCR1}. Now, the Gray map $\psi_1$ on $R_1$ that we have defined is the same map used in these works that we have mentioned and the homogeneous weight coincides with the Lee weight used. That is why, we will focus on $(1+u)$-QT codes here. We will fix the index at $3$. The results that we have found are quite different than the ones found in the above-mentioned works. We tabulate our results in the following table. To save space, we will replace $1+u$ by $3$. A typical generator of the QT-code will be given in the form $(g_1g_2\dots g_m|h_1h_2\dots h_m|r_1r_2\dots r_m)$. Thus for example $(1u|30|u3)$ will denote the one generator $QT$-code generated $(g(x),h(x),r(x))$ where $g(x)=1+ux$, $h(x)=1+u$ and $r(x) = u+(1+u)x$.

\begin{table}[H]
\caption{Optimal binary codes from $(1+u,3)$-QT codes over $R_1$ of length $3m$}
\label{tab:new96}
\begin{center}
\begin{tabular}{|c|l|c|}
\hline
$m$ & Generator of the code & Binary Image under $\psi_1$ \\ \hline
$2$ & $(0u|0u|uu)$
& $[12,2,8]$ \\ \hline
$2$ & $(10|11|3u)$
& $[12,4,6]$ \\ \hline
$2$ & $(0u|33|13)$
& $[12,3,6]$ \\ \hline
$3$ & $(00u|011|u33)$
& $[18,5,8]$ \\ \hline
$3$ & $(00u|111|111)$
& $[18,4,8]$ \\ \hline
$3$ & $(001|113|1u1)$
& $[18,6,8]$ \\ \hline
$3$ & $(0uu|0uu|uu0)$
& $[18,2,12]$ \\ \hline
$4$ & $(0011|001u|00u1)$
& $[24,8,8]$ \\ \hline
$4$ & $(0011|0013|1u1u)$
& $[24,7,10]$ \\ \hline
$4$ & $(000u|00uu|0uuu)$
& $[24,4,12]$ \\ \hline
$4$ & $(0u0u|0u0u|uuuu)$
& $[24,2,16]$ \\ \hline
$5$ & $(0011u|001u3|00u33)$
& $[30,8,12]$ \\ \hline
$5$ & $(001u1|0013u|01111)$
& $[30,9,12]$ \\ \hline
$5$ & $(13131|uuuuu|13131)$
& $[30,2,20]$ \\ \hline
$6$ & $(uuuu11|uuu103|u1u311)$
& $[36,11,12]^b$ \\ \hline
$6$ & $(u1u103|u10101|113133)$
& $[36,6,16]$ \\ \hline
$6$ & $(u1u3u1|010301|133113)$
& $[36,4,18]$ \\ \hline
$6$ & $(0u0u0u|0u0u0u|uuuuuu)$
& $[36,2,24]$ \\ \hline
$7$ & $(uuu1013|uu01033|uu11101)$
& $[42,11,16]$ \\ \hline
$7$ & $(uu10333|u1330u1|u03u331)$
& $[42,6,20]$ \\ \hline
$7$ & $(1313131|uuuuuuu|1313131)$
& $[42,2,28]$ \\ \hline
\end{tabular}
\end{center}
\end{table}

\begin{remark}
The $[36,11,12]$-code that is marked with $^b$ is the best known code of these parameters. The theoretical upper bound for the minimum distance of the code of length $36$ and dimension $11$ is $13$, which has not been attained yet. All the rest of the codes in the table are optimal, meaning that they attain the theoretical upper bounds. We also note that all the binary codes given in the above table are equivalent to $6$-QC codes.
\end{remark}

A pattern in the table has led us to observe the following:
Suppose $m$ is even. Then, since $(1+u)\cdot u = u$, the $(1+u,3)$-QT code $C$ generated by $(0u0u \dots 0u|0u0u\dots 0u|uu \dots u)$ contains just four codewords given by
$$(00\dots 0|00\dots 0|00\dots 0), (0u0u \dots 0u|0u0u\dots 0|uu \dots u),$$
$$ (u0u0 \dots u0|u0u0\dots u0|uu \dots u), (uu \dots u|uu\dots u|00 \dots 0).$$
The homogeneous weight distribution of this code is given by $1+3z^{4m}$. Thus in the $\psi_1$ image we get a binary $[6m,2,4m]$-code.

When $m$ is odd, we take the generator to be of the form $(1313 \dots 1|1313\dots 1|uu \dots u)$. Remembering that $3$ stands for $1+u$ and that $(1+u)(1+u) = 1$ and $1+(1+u)=u$, again $C$ has 4 codewords in this case, given by
$$(00\dots 0|00\dots 0|00\dots 0), (1313 \dots 1|1313\dots 1|uu \dots u)$$
$$(3131 \dots 3|3131\dots 3|uu \dots u), (uu \dots u|uu\dots u|00 \dots 0).$$
This code also has minimum homogeneous weight $4m$.

Thus we have found the following result:
\begin{theorem}
For any $m\geq 1$, we can obtain a binary code of parameters $[6m,2,4m]$ as the $\psi_1$-image of a $(1+u,3)$-QT code over $R_1$. Note that all such codes attain the Griesmer bound and thus are optimal.
\end{theorem}

\subsection{Optimal binary codes from QT codes over $R_2$}
We give some examples of optimal binary codes from the $\psi_2$-images of cyclic and QC-codes over $R_2$. $R_2=\F_2+u\F_2+v\F_2+uv\F_2$, having $16$ elements, we use the hexadecimal notation to denote the elements of $R_2$ in numerical form. For this we fix the basis $\{uv,v,u,1\}$ for $R_2$ and represent the element $auv+bv+cu+d$ as the $4$-bit $abcd$ which, then is represented by its hexadecimal value. Thus for example $uv+u+1$ is represented by $b$, which stands for the numerical value of $11$, while $v+u+1$ is represented by $7$.

\begin{table}[H]
\caption{Optimal binary codes from $\psi_2$-images of cyclic codes over $R_2$}
\label{cycR2}
\begin{center}
\begin{tabular}{|c|l|c|c|}
\hline $n$ & Generator of the code & Binary Image under $\psi _{2}$ & As 8-QC in Database in \cite{ChenDatabase}
\\ \hline $2$ & $(11)$ & $[16,4,8]$ & New \\ \hline $3$ & $(088)$ &
$[24,2,16]$ & New \\ \hline $3$ & $(246)$ & $[24,4,12]$ & New \\ \hline $3$ &
$(135)$ & $[24,8,8]$ & New \\ \hline $3$ & $(019)$ & $[24,9,8]$ & New \\ \hline
$4$ & $(0282)$ & $[32,4,16]$ & --- \\ \hline $4$ & $(1199)$ & $[32,5,16]$ & New
\\ \hline $4$ & $(1119)$ & $[32,6,16]$ & New \\ \hline $5$ & $(11111)$ &
$[40,4,20]$ & New \\ \hline $5$ & $(02442)$ & $[40,8,16]$ & New\\ \hline $6$ &
$(aec26c)$ & $[48,4,24]$ & New\\ \hline $6$ & $(088088)$ & $[48,2,32]$ & New \\
\hline $7$ & $(0888008)$ & $[56,3,32]$ & New \\ \hline $8$ & $(ceec4e6c)$
& $[64,6,32]$ & ---\\ \hline
\end{tabular}
\end{center}
\end{table}

\begin{table}[H]
\caption{Optimal binary codes from $\psi_2$-images of $\ell$-QC codes of length $\ell\cdot m$ over $R_2$ } \label{2QCR2}
\begin{center}
\begin{tabular}{|c|c|l|c|c|}
\hline $\ell$ & $m$ & Generator of the code & Binary Image under $\psi _{2}$ & As $\ell$-QC in Database in \cite{ChenDatabase}
\\ \hline $2$& $2$ & $(2c|2c)$ & $[32,4,16]$ & New 16-QC\\
\hline  $2$ & $2$ & $(5F|57)$ & $[32,5,16]$ & New 16-QC\\
\hline $2$ & $3$ & $(088|088)$ & $[48,2,32]$ & New 16-QC \\
\hline  $2$& $3$ & $(066|6e8)$ & $[48,4,24]$ & New 16-QC\\
\hline $2$& $3$ & $(246|2c6)$ & $[48,5,24]$ & New 16-QC \\
\hline $2$& $4$ & $(aaa2|4e4e)$ & $[64,5,32]$ & New 16-QC\\
\hline $2$& $4$ & $(1573|bf51)$ & $[64,6,32]$ & New 16-QC\\
\hline $2$& $4$ & $(f539|b579)$ & $[64,7,32]$ & New 16-QC\\
\hline $3$ & $2$ & $(08|08|88)$ & $[48,2,32]$ & ---\\
\hline $3$ & $3$ & $(231|f87|bc7)$ & $[72,8,32]$ & New 24-QC\\
\hline
\end{tabular}
\end{center}
\end{table}

\begin{remark}
The codes in Tables 1-3 have different properties even if some of them have the same parameters. The binary codes in Table 1 are equivalent to $6$-QC codes while the binary codes in Table 2 all are 8-QC codes and the ones in Table 3 are 16-QC or 24-QC according as $\ell=2$ or $3$.
\end{remark}
\section{Conclusion}
It is known in the literature of coding theory that quasicyclic codes (and more generally quasitwisted codes) tend to have good parameters as they satisfy a modified version of the Gilbert-Varshamov bound. Many optimal codes have constructions that are related to quasicyclic codes.

In this work, we demonstrated that, considering QT codes over $R_k$ with respect to the homogeneous weight lead to many optimal (self-orthogonal) codes. Our constructions have resulted much more than what has been put into Tables 1-3 (with many different weight enumerators and automorphism groups), however, we have given a sample of these to illustrate the effectiveness of our constructions. We have found many new quasiyclic codes that can be added to the database of known binary qausicyclic codes in \cite{ChenDatabase}. We have also come up with rather simple constructions for many optimal binary codes in the database \cite{Grassl}, that otherwise have rather complicated constructions. We believe such an approach can be applied to other Frobenius rings as well.

\bibliographystyle{amsplain}

\end{document}